\documentclass[preprint,1p,11pt]{StyleFiles/ISAS_IR}

\pdfoutput=1

\usepackage[latin1]{inputenc}
\usepackage{calc}
\usepackage{stfloats} 
\usepackage{calc}
\usepackage{microtype}
\usepackage[intlimits,sumlimits,namelimits]{amsmath} 
\everymath{\displaystyle} 
\usepackage{amssymb}
\usepackage{amsfonts}
\usepackage{psfrag}
\usepackage{verbatim}
\usepackage{multirow}
\usepackage{multicol}
\usepackage{algorithm2e}
\usepackage{blkarray}
\usepackage{lmodern}
\usepackage{wrapfig}
\usepackage{paralist}
\usepackage{trfsigns}
\usepackage{version}
\usepackage{wasysym}
\usepackage{bm}
\usepackage{color}
\usepackage{xcolor}
\usepackage{graphicx}
\usepackage{asymptote}
\usepackage[bigfiles]{media9}
\usepackage[export]{adjustbox} 
\usepackage{caption}
\def\vec#1{\underline{#1}}
\def\mat#1{{\mathbf #1}}

\def\1_2{{\frac{1}{2}}}

\DeclareMathOperator{\erf}{erf}
\DeclareMathOperator{\erfinv}{erfinv}

\def\rv#1{\boldsymbol{#1}}

\DeclareMathOperator{\diag}{diag}
\DeclareMathOperator{\E}{E}
\DeclareMathOperator{\HeaviSide}{H}

\def\given{\, | \,}
\newcommand{\pluseq}{\mathrel{+}=}

\newcommand{\diveq}{\mathrel{/}=}
\def\ge{\geqslant}

\def\DeltaX#1{\hat{\vec{\Delta}}_{#1}}

\def\d{\mathrm{d}}

\def\NewR{\mathbb{R}} 

\def\NewS{\mathbb{S}}

\def\Eq#1{(\ref{#1})}

\def\Sec#1{Sec.~\ref{#1}}
\def\SubSec#1{Subsec.~\ref{#1}}

\def\Fig#1{Fig.~\ref{#1}}

\def\Box#1{{\fboxsep2pt\fbox{$#1$}}}

\newlength\EqLen

\def\ScaleInner#1{%
\settowidth{\EqLen}{#1}
\ifdim\EqLen < \columnwidth%
  \begin{equation*}%
    \begin{minipage}{\EqLen}#1\end{minipage}%
  \end{equation*}%
\else%
  \begin{equation*}%
    \resizebox{0.99\columnwidth}{!}{\begin{minipage}{\EqLen}#1\end{minipage}}%
  \end{equation*}%
\fi%
}%

\def\Scale#1
  {
  \ScaleInner{$ #1 $} 
  }
  
\def\LongVersion#1{}

\def\citep#1{(\cite{#1})}

\usepackage{amsthm}

\makeatletter
\renewenvironment{proof}[1][\proofname] {\par\pushQED{\qed}\normalfont\topsep6\p@\@plus6\p@\relax\trivlist\item[\hskip\labelsep\bfseries\slshape#1\@addpunct{.}]\ignorespaces}{\popQED\endtrivlist\@endpefalse}
\makeatother

\newtheorem{theoremx}{Theorem}[section]
\newenvironment{theorem}
  {\pushQED{\qed}\theoremx}
  {\popQED\endtheoremx}

\newtheorem{remarkx}{Remark}[section]
\newenvironment{remark}
  {\pushQED{\qed}\remarkx}
  {\popQED\endremarkx}

\newtheorem{examplex}{Example}[section]
\newenvironment{example}
  {\pushQED{\qed}\examplex}
  {\popQED\endexamplex}

\makeatletter
\newcommand\SaveEquation[2]{\@namedef{equation@#1}{#2}}
\newcommand\UseEquation[1]{\@nameuse{equation@#1}}
\makeatother

\SetKw{KwBy}{by}

\SetCommentSty{mycommfont}

\date{}

\begin{document}

\begin{frontmatter}

\title{Deterministic Sampling of Multivariate Densities\\
based on Projected Cumulative Distributions
}

\author{Uwe~D.~Hanebeck}

\address{%
	Intelligent Sensor-Actuator-Systems Laboratory (ISAS)\\
	Institute for Anthropomatics and Robotics\\
	Karlsruhe Institute of Technology (KIT), Germany\\
	email: Uwe.Hanebeck@kit.edu
	}%

\begin{abstract}
	%
%
We want to approximate general multivariate probability density functions by deterministic sample sets.
%
%
For optimal sampling, the closeness to the given continuous density has to be assessed. This is a difficult challenge in multivariate settings. Simple solutions are restricted to the one-dimensional case.
%
%
In this paper, we propose to employ one-dimensional density projections.
%
%
These are the Radon transforms of the densities.
%
%
For every projection, we compute their cumulative distribution function.
%
%
These Projected Cumulative Distributions (PCDs) are compared for all possible projections (or a discrete set thereof). This leads to a tractable distance measure in multivariate space.
%
%
The proposed approximation method is efficient as calculating the distance measure mainly entails sorting in one dimension. It is also surprisingly simple to implement.

\end{abstract}

\end{frontmatter}

\section{Introduction} \label{Sec_Intro}

    %
%
Approximating a given probability density function by another one is a ubiquitous problem.
%
%
It arises when the given density is too complex for further processing. This could be caused by an undesired and complicated functional representation. Even if the representation is of a desired form, the number of parameters could be unnecessarily large, e.g., Gaussian mixtures.
%
%
Hence, it has to be replaced by a more convenient representation or one with less parameters.

%
%
Recursive Bayesian filtering is an important use case for density approximation. For nonlinear system models, the prediction step is usually infeasible for continuous prior densities. Hence, an approximation of the prior by a set of samples is often used instead. The filtering step entails the product of the prior and a likelihood function. In some cases, e.g., Gaussian mixture densities, this product leads to an undesired exponential increase in parameters and calls for regular reapproximation.

%
%
The given density and its approximation can both be either continuous or discrete (over continuous domains).
%
%
Four cases can be distinguished:
\begin{itemize}
    \item [\bfseries 1) Continuous/continuous:] Approximation of a given continuous density by another continuous density of different type or of the same type with a different number of parameters. This is called density reapproximation.
    \item [\bfseries 2) Continuous/discrete:] Approximation of a given continuous density by a discrete density. This is called sampling. \emph{This is the case we consider in this paper.}
    \item [\bfseries 3) Discrete/continuous:] Approximation of a given discrete density, i.e, samples, by a continuous density. This is called density estimation.
    \item [\bfseries 4) Discrete/discrete:] Approximation of a given discrete density, i.e, samples, by another discrete density, i.e., another sample set. This is called Dirac mixture reapproximation. An important special case is the approximation with a smaller number of samples. This is called Dirac mixture reduction.
\end{itemize}

%
%
In this paper, we focus on the second case of approximating a given continuous density by a deterministic discrete density, see \Fig{Teaser_DMA_GM}.
%
%
The discrete density is represented by a weighted set of Dirac delta functions at specific locations.
%
%
The deterministic sampling procedure includes the calculation of appropriate locations and, if desired, also appropriate weights.
%
%
In this paper we focus on equally weighted samples or at least on samples with prescribed weights.

\section{Problem Formulation} \label{Sec_ProbForm}

	%
%
The problem we solve in this paper is the systematic approximation of a given multivariate continuous density $\tilde{f}(\vec{x})$ by a Dirac mixture density $f(\vec{x})$, i.e., by a set of $L$ deterministic samples, given by
\begin{equation}
    f(\vec{x}) = f(\vec{x} \given \hat{\mat{X}}) 
    = \sum_{i=1}^L w_i \, \delta(\vec{x} - \hat{\vec{x}}_i)
\label{Eq_Dirac}
\end{equation}
with given weights $w_i \in \NewR_{>0}$, $\textstyle\sum_{i=1}^L w_i = 1$, and desired locations $\hat{\vec{x}}_i \in \NewR^N$. $N$ is the number of dimensions. The locations are collected in an $N \times L$ matrix $\hat{\mat{X}} = \begin{bmatrix} \hat{\vec{x}}_1, \hat{\vec{x}}_2, \ldots, \hat{\vec{x}}_L \end{bmatrix}$.

%
%
As we insist on a systematic approximation, we require a rigorous distance measure between the given density $\tilde{f}(\vec{x})$ and its approximation $f(\vec{x})$ to be minimized. As this entails a comparison between a continuous and a discrete density, finding an appropriate and feasible distance  measure is a major challenge.
%
%
An overview of the current literature is given in the next section.

\if{false}
%
%
An additional aspect is the use of the dual approximation, i.e., the approximation of a Dirac mixture by a Gaussian mixture.
%
%
As a rigorous distance measure between the Dirac mixture and its approximation is minimized, this avoids singularity effects that are known to haunt maximum likelihood based approaches.
\fi

\begin{figure*}[t]
\includegraphics[width=100mm,valign=m]{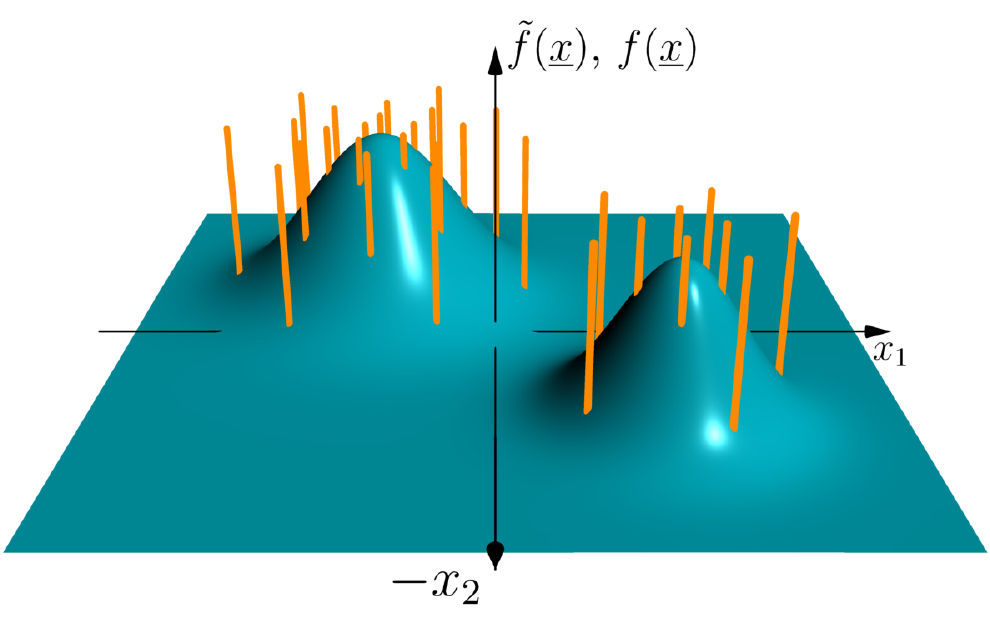}
\hspace*{\fill}
\includegraphics[width=65mm,valign=m]{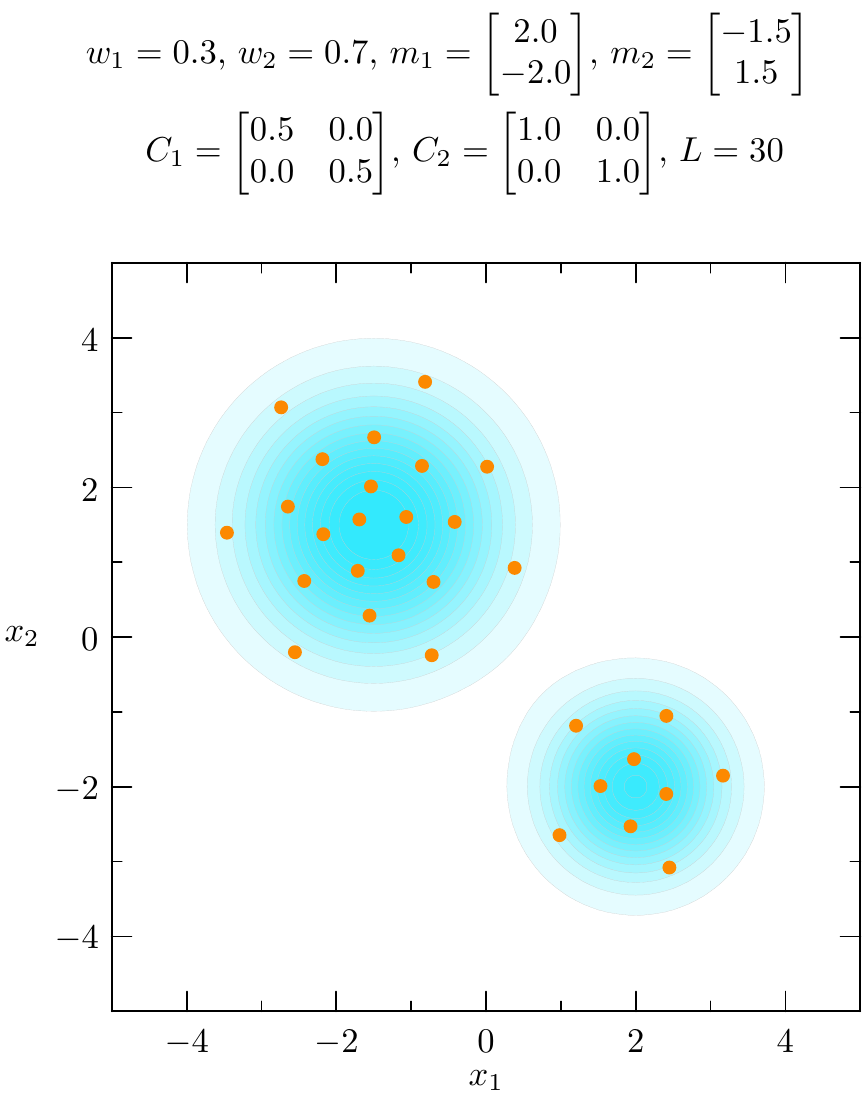}
\caption{Dirac mixture approximation of Gaussian mixture with two components. Left: Perspective view. Right: Top view and parameters.}
\label{Teaser_DMA_GM}
\end{figure*}

\section{State of the Art} \label{Sec_State_of_Art}

	%
%
We will now give an overview of the current state of the literature on approximating given continuous probability density functions by a set of samples. 
%
%
Most results are available for Gaussian densities. Although we focus on the general non-Gaussian case, we will review some literature for Gaussians.
%
%
We will focus on deterministic approximation. However, a few words on random sampling will be given first.

%
%
Generating samples for a given probability density function is usually performed by random sampling. This is the simplest and fastest method and well-established algorithms are available for that purpose. However, with random sampling, samples are produced independently. As a result, convergence of its statistics, such as the moments, to the true values is slow so that many samples are required. The independent generation of samples also leads to a non-homogeneous coverage of the given density.

%
%
\subsection{Deterministic Sampling: Moment-based Approaches for the Gaussian Case}
As an alternative to random sampling, deterministic sampling methods have been proposed. Moment-based approximations of Gaussian densities are the basis for Linear Regression Kalman Filters (LRKFs), see \cite{lefebvreLinearRegressionKalman2005}. Examples are the Unscented Kalman Filter (UKF) in \cite{julierNewMethodNonlinear2000} and its scaled version in \cite{julierScaledUnscentedTransformation2002}. A higher-order generalization is proposed in \cite{tenneHigherOrderUnscented2003}. These methods are limited to a fixed number of samples per dimension. An arbitrary number of deterministic samples of a Gaussian density is provided by the spherical-radial rule in \cite{arasaratnamCubatureKalmanFilters2009} for the so called cubature filter.

%
%
\subsection{Deterministic Sampling: Distance-based Approaches}
%
%
Instead of relying on moment relations that quickly become intractable, distance-based approaches minimize an appropriate distance measure between the given continuous density and its discrete approximation.
%
%
Distance-based approaches have several distinct advantages: (1)~Arbitrary number of samples. (2)~Homogeneous allocation of samples. (3)~Prioritization of density regions especially important for further processing.

%
%
\subsubsection*{One-dimensional Densities}
For one-dimensional (1D) densities, 
%
%
the vast literature on 1D statistical tests can be used.
%
%
These are used for testing the hypothesis that a set of samples has a certain density. Inverting these tests can be used to fit samples to a given density.
%
%
We will focus on one-sample tests that compare the cumulative distribution function (CDF) of the continuous density with the empirical distribution function (EDF) of the samples. The Kolmogorov-Smirnov test 
\cite{masseyKolmogorovSmirnovTestGoodness1951} uses the maximum between the CDF and the EDF, which is difficult to optimize. Integral squared distances between CDF and EDF were first considered by Cram\'{e}r \cite{cramerCompositionElementaryErrors1928} and von Mises leading to the Cram\'{e}r-von Mises test. Several tests have been developed on that basis, differing only in the weighting function for the squared difference. One notable test is the Anderson-Darling test \cite{andersonAsymptoticTheoryCertain1952}.

%
%

Dirac mixture approximation (with equally weighted samples) of continuous densities based on the Cram\'{e}r-von Mises criterion for the univariate case is proposed in \cite{CDC06_Schrempf-DiracMixt}. Optimizing both weights and locations is considered in \cite{MFI06_Schrempf-CramerMises}. Sequentially increasing the number of samples is the focus of \cite{CDC07_HanebeckSchrempf} and is applied to recursive nonlinear prediction in \cite{ACC07_Schrempf-DiracMixt}.

%
%
\subsubsection*{Multivariate Densities}
%
%
\emph{\bfseries Challenges:}
For $N$-dimensional (ND) densities, with $N>1$, systematic approximations are much harder to achieve as
%
%
cumulative distribution functions in higher dimensions are difficult to calculate. In addition, they are not unique \cite{MFI08_Hanebeck-LCD}. This is in contrast to the univariate case, where there are only two directions of integration for obtaining the cumulative distribution function from the density. The results of the two directions are dependent, i.e., sum up to one. In higher dimensions, the number of orderings of possible integration directions increases exponentially with $2^N$, where $N$ is the number of dimensions. $2^N-1$ of these distributions are independent 
\cite[p.~617]{peacockTwoDimensionalGoodnessofFitTesting1983}.
%
%
Selecting an arbitrary distribution from these possibilities for density approximation typically results in a bias of the estimated parameters \cite{MFI08_Hanebeck-LCD}.

%
%
\emph{\bfseries Consider All CDF Orderings:}
In order to cope with the non-uniqueness, it has been proposed in the context of multivariate statistical tests, to simply consider all possible orderings 
\cite[p.~617]{peacockTwoDimensionalGoodnessofFitTesting1983}. This can easily be done in two or three dimensions \cite{fasanoMultidimensionalVersionKolmogorovSmirnov1987}. For higher-dimensional spaces, considering all ordering becomes impractical.

%
%
\emph{\bfseries Localized Cumulative Distribution:}
Instead of using the standard cumulative distribution with its non-uniqueness issues, an alternative cumulative distribution is introduced in \cite{MFI08_Hanebeck-LCD}. It is called Localized Cumulative Distribution (LCD). A distance measure is then defined based on the LCDs of the two densities to be compared.
%
%
The LCD-based approach is used for systematically approximating arbitrary multi-dimensional Gaussian densities \cite{CDC09_HanebeckHuber}. An improved method with better numerical stability that exploits the symmetry of the Gaussian density with symmetric samples is proposed in \cite{JAIF16_Symmetric_S2KF_Steinbring}. For multivariate standard normal distributions, a more efficient scheme is presented that relies on a subsequent transformation \cite{ACC13_Gilitschenski}.
%
%
The LCD-based approach provides very good approximations. However, for the case of approximating continuous densities with Dirac mixtures, closed-form solutions for the distance measure are only available in special cases.

%
%
\emph{\bfseries Repulsion Kernels and Blue Noise:}
%
%
%
Comparing cumulative distribution functions is, of course, not the only way to compare discrete densities or continuous and discrete densities. Alternatives include comparing characteristic functions~\cite{stephensEDFStatisticsGoodness1974} and using kernel estimates.
%
%
Kernel estimates of a discrete density are exceptionally simple in the considered case of finding the best-fitting discrete density for a given continuous density. In that case, we use the fact that the density at a certain sample location is equal to the corresponding value of the continuous density~\cite{fattalBlueNoisePointSampling2011}. So called repulsion kernels are employed at the sample locations. The induced kernel density is then compared to the given continuous density by means of, e.g., an integral squared distance~\cite{CISS14_Hanebeck}. For given Gaussian densities that are approximated by Dirac mixtures, a closed-form expression for the distance measure is derived in~\cite{Fusion14_Hanebeck}. A randomized optimization method is then used for finding the optimal locations instead of a quasi-Newton method in~\cite{CISS14_Hanebeck}.

%
%
\emph{\bfseries Reduction to Univariate Case:}
We have seen that the univariate case is significantly simpler than the multivariate one. So it comes as no surprise that many attempts have been made to reduce the multivariate case to the univariate one:
\begin{itemize} \itemsep0pt
    \item \emph{Approximation on Principal Axes:} The first idea that immediately springs to mind is to use 1D approximations on the principal axes of the given continuous density. This is, of course, limited to densities where principal axes can naturally be defined. This includes Gaussian densities \cite{IFAC08_Huber} in $\NewR^N$ and the Bingham distribution \cite{ECC19_Li} defined on $\mathbb{S}^{N-1} \subset \NewR^N$.
    \item \emph{Cartesian Products:} For random vectors with independent components, e.g., axis-aligned Gaussian densities, the $N$ marginal densities can be approximated independently. The $N$-dimensional Dirac mixture approximation can then be obtained by a Cartesian product of the individual marginal approximations. An obvious disadvantage is that the total number of samples is the product of the individual numbers of samples, i.e., scales exponentially with the number of dimensions $N$.
    \item \emph{One-dimensional Projections:} General densities cannot be represented by their marginals, only densities that can be factorized. However, instead of limiting projections to the coordinate axes as done for the marginals, more projections onto different axes can be considered. For representing arbitrary densities, all possible one-dimensional projections have to be considered. This is called the Radon transform of the given density, see the original publication in German \cite{radonUberBestimmungFunktionen1917}, its English translation
    \cite{radonDeterminationFunctionsTheir1986}, and
    \cite{deansRadonTransformHigher1978} for the higher-dimensional case. In \cite{bonneelSlicedRadonWasserstein2015}, a multivariate Wasserstein distance is constructed from the Wasserstein distance between one-dimensional projections (called slices therein).
\end{itemize}

%
%
\subsection{Combined Random and Deterministic Sampling}
Besides pure random or pure deterministic sampling, combined random and deterministic sampling has been proposed. The idea is to perform a random rearrangement, e.g., by rotation, of a deterministic set of sample points, see \cite{strakaRandomizedUnscentedKalman2012}.

\section{Key Idea and Main Results} \label{Sec_KeyIdea}

	%
%
As we have seen in the state of the art in the previous section: Calculating the distance between a continuous and a discrete density is a difficult challenge in the multivariate case.
%
%
However, calculating the distance between univariate densities is feasible.

\paragraph*{Na\"ive Approach}
%
%
We now assume an integral distance measure of the form
\begin{equation*}
    D_N = D_N\left( \tilde{f}, f \right) 
    = \int_{\NewR^N} d_N\left( \tilde{f}(\vec{x}), f(\vec{x}) \right)  \, \d \vec{x} \enspace ,
\end{equation*}
where $d_N(.,.)$ is a distance measure between two densities that is integrated over the entire space.
%
%
For example, $d_N(.,.)$ could be defined in terms of the cumulative distribution functions of $\tilde{f}(\vec{x})$ and $f(\vec{x})$ given by $\tilde{F}(\vec{x})$ and $F(\vec{x})$ as
\begin{equation*}
    d_N\left( \tilde{f}(\vec{x}), f(\vec{x}) \right) 
    = \left[ \tilde{F}(\vec{x}) - F(\vec{x}) \right]^2 \enspace .
\end{equation*}
The cumulative distribution function $F(\vec{x})$ is given by
\begin{equation*}
    F(\vec{x}) = \int_{-\vec{\infty}}^{\vec{x}} f(\vec{t}) \, \d \vec{t} \enspace .
\end{equation*}
An analogous result holds for $\tilde{F}(\vec{x})$, so the distance measure is
\begin{equation*}
    D_N\left( \tilde{f}, f \right) 
    = \int_{\NewR^N} \left[ \tilde{F}(\vec{x}) - F(\vec{x}) \right]^2 \, \d \vec{x} \enspace .
\end{equation*}
However, as we have seen in the last section, multivariate cumulative distributions are not unique and difficult to calculate.

%
%
One option to solve this problem is to use the Localized Cumulative Distribution (LCD) instead of the standard cumulative distribution. The LCD is unique, easy to calculate, and gives high-quality results. However, the complexity depends on the number of dimensions considered.

\paragraph*{Proposed Approach}
%
%
Here, we pursue a different idea.
\begin{enumerate}
\item We represent the two densities $\tilde{f}(\vec{x})$ and $f(\vec{x}) = f(\vec{x} \given \hat{\mat{X}})$ from \Eq{Eq_Dirac} by the infinite set of one-dimensional projections onto all unit vectors $\vec{u} \in \mathbb{S}^{N-1}$. For $\tilde{f}(\vec{x})$, we obtain the set of projections $\tilde{f}(r \given \vec{u})$ with $r = \vec{u}^\top \vec{x}$, which is the Radon transform of $\tilde{f}(\vec{x})$. In the same way, we obtain the Radon transform $f(r \given \hat{\vec{r}}, \vec{u})$ of the Dirac mixture approximation in \Eq{Eq_Dirac}, where $\hat{r}_i = \vec{u}^\top \hat{\vec{x}}_i$ are the transformed Dirac locations for $i \in \{1, 2, \ldots, L\}$.
\item We compare the one-dimensional projections for every $\vec{u} \in \mathbb{S}^{N-1}$. For the comparison, we can use the univariate cumulative distribution functions $\tilde{F}(r \given \vec{u})$ and $F(r \given \hat{\vec{r}}, \vec{u})$ as these are unique, well defined, and easy to calculate. The resulting distance measures depend on the projection vector $\vec{u}$ and are given by
\begin{equation*}
    D_1(\vec{u}) 
    = \int_\NewR d_1\left( \tilde{f}(r \given \vec{u}), 
    f(r \given \hat{\vec{r}}, \vec{u}) \right) \, \d r \enspace ,
\end{equation*}
with
\begin{equation*}
    d_1\left( \tilde{f}(r \given \vec{u}), f(r \given \hat{\vec{r}}, \vec{u}) \right) 
    = \left[ \tilde{F}(r \given \vec{u}) - F(r \given \hat{\vec{r}}, \vec{u}) \right]^2 \enspace .
\end{equation*}
\item We can easily minimize this one-dimensional distance measure, e.g., by Newton's method, which yields changes $\DeltaX{r}(\vec{u})$ of the given projected Dirac locations $\hat{\vec{r}}(\vec{u})$ for every projection $\vec{u} \in \mathbb{S}^{N-1}$.
\item These changes can be projected back to the Dirac locations as $\DeltaX{x,i}(\vec{u})$
for $i \in \{1, 2, \ldots, L\}$ for every projection $\vec{u} \in \mathbb{S}^{N-1}$.
\item We integrate the one-dimensional distance measures over all possible unit vectors $\vec{u}$, which gives a multivariate distance measure
\begin{equation*}
    D_N = D_N\left( \tilde{f}, f \right)
    = \frac{1}{A_N} \int_{\mathbb{S}^{N-1}} D_1(\vec{u})  \, \d \vec{u}
    = \frac{1}{A_N} \int_{\mathbb{S}^{N-1}} 
    \int_\NewR \left[ \tilde{F}(r \given \vec{u}) - F(r \given \hat{\vec{r}}, \vec{u}) \right]^2 \, \d r \, \d \vec{u}
    \enspace ,
\end{equation*}
where $A_N$ is the surface area of the (hyper)sphere $\NewS^{N-1}$ embedded in $\NewR^N$.
\item The changes from step 4 can be averaged over all $\vec{u} \in \mathbb{S}^{N-1}$
\begin{equation*}
    \DeltaX{x,i} = \frac{1}{A_N} \int_{\mathbb{S}^{N-1}} 
    \DeltaX{x,i}(\vec{u}) \, \d \vec{u} \enspace ,
\end{equation*}
which corresponds to a combined Newton step for all $\vec{u}$ for $\hat{\vec{x}}_i$, $i \in \{1, 2, \ldots, L\}$.
\end{enumerate}

\section{Density Representation via One-dimensional Projections (Radon Transform)} \label{Sec_Radon}

	%
%
In this section, we will represent general $N$-dimensional probability density functions via the set of all one-dimensional projections. This is also known as the Radon Transform \cite{radonUberBestimmungFunktionen1917}\footnote{English translation available \cite{radonDeterminationFunctionsTheir1986}.}.

\begin{figure*}[t]
    \begin{center}
    \includegraphics{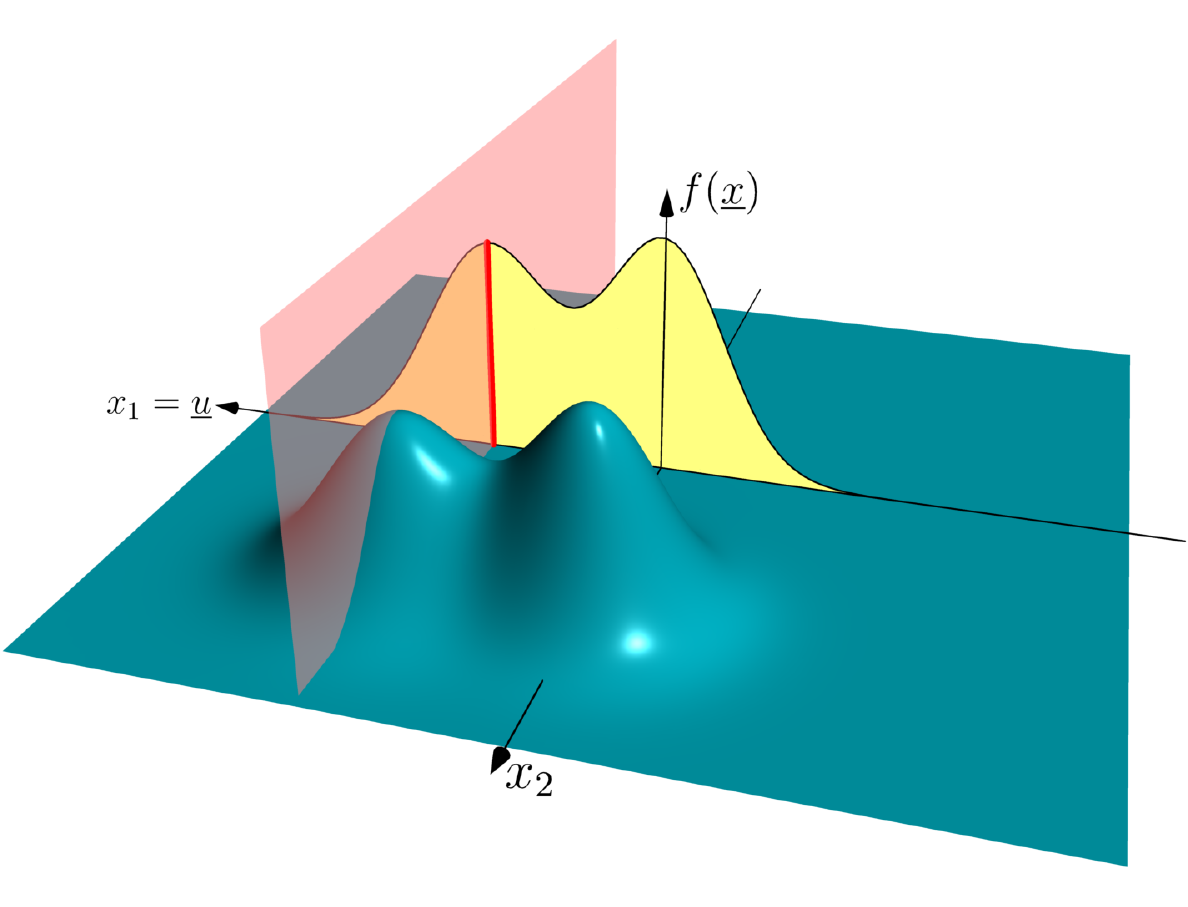}
    \end{center}
    \caption{Visualization of Radon transform. Blue: Original multivariate density $f(\vec{x})$. Red: Single line integral over $f(\vec{x})$. Yellow: Projected density $f_r(r \given \vec{u})$ with $\vec{u} = [1, 0]^\top$.}
    \label{Radon_Visualize_2}
\end{figure*}

%
%
We consider the linear projection of a random vector $\rv{\vec{x}} \in \NewR^N$ to a scalar random variable $\rv{r} \in \NewR$
\begin{equation}
    \rv{r} = \vec{u}^\top \rv{\vec{x}}
\label{Eq_ProjectRV}
\end{equation}
on the line described by the unit vector $\vec{u} \in \mathbb{S}^{N-1}$.

%
%
Given the probability density function $f(\vec{x})$ of the random vector $\rv{\vec{x}}$, the density $f_r(r \given \vec{u})$ of $\rv{r}$ is given by
\begin{equation}
    f_r(r \given \vec{u}) = \int_{\NewR^N} f(\vec{t}) \, \delta(r - \vec{u}^\top 
    \vec{t}) \, \d \vec{t} \enspace .
\label{Eq_Radon}
\end{equation}
$f_r(r \given \vec{u})$ is the Radon transform of $f(\vec{x})$ for all $\vec{u} \in \mathbb{S}^{N-1}$. 
\begin{example}[2D Radon transform] 
    We consider a two-dimensional Gaussian mixture with two components. A single projection of this density is visualized for $\vec{u} = [1, 0]^\top$ in \Fig{Radon_Visualize_2} and for $\vec{u} = [-1, 1]^\top/\sqrt{2}$ in \Fig{Radon_Visualize_1}.
\end{example}

\begin{figure*}[t]
    \begin{center}
    \includegraphics{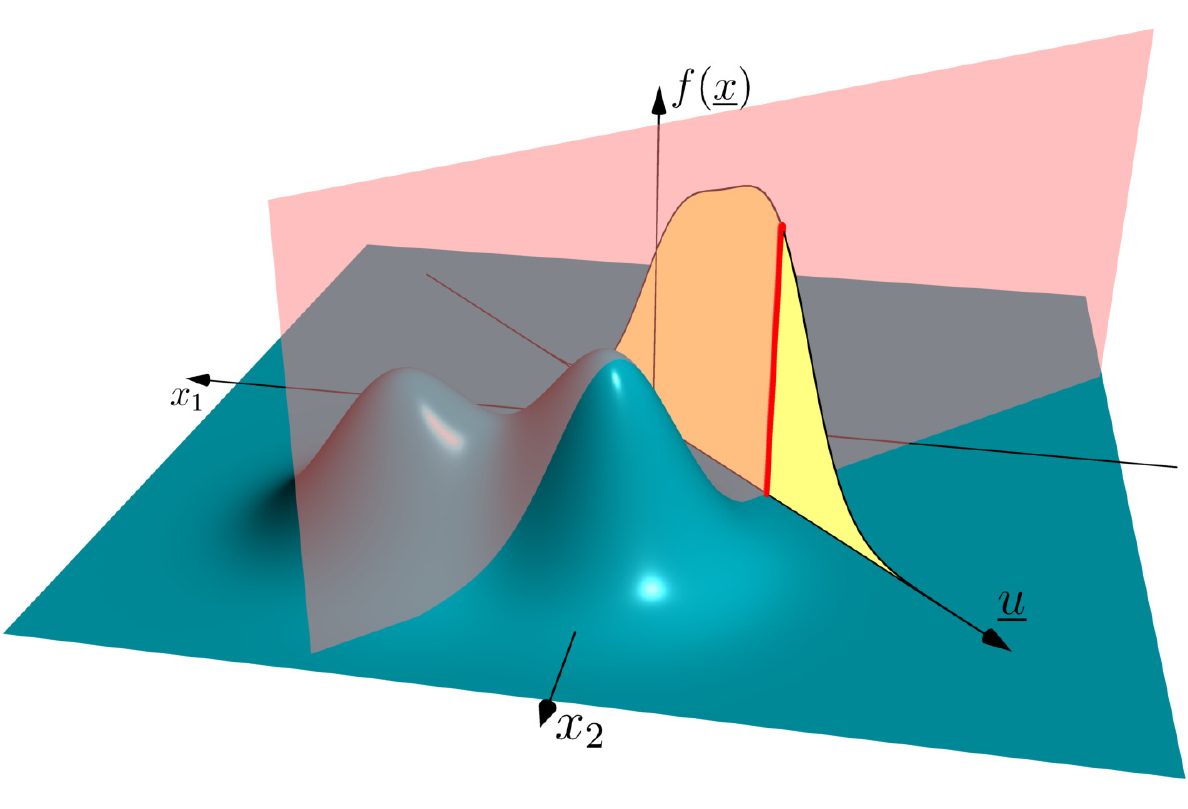}
    \end{center}
    \caption{Visualization of Radon transform. Blue: Original multivariate density $f(\vec{x})$. Red: Single line integral over $f(\vec{x})$. Yellow: Projected density $f_r(r \given \vec{u})$ with $\vec{u} = [-1, 1]^\top/\sqrt{2}$.}
    \label{Radon_Visualize_1}
\end{figure*}

%
%
\paragraph*{Dirac Mixture Densities}
We consider the Dirac mixture approximation $f(\vec{x} \given \hat{\mat{X}})$ in \Eq{Eq_Dirac}.
\begin{theorem}
The Radon transform of $f(\vec{x} \given \hat{\mat{X}})$ is given by
\begin{equation*}
    f_r(r \given \hat{\vec{r}}, \vec{u}) 
    = \sum_{i=1}^L w_i \, \delta\big( r - \hat{r}_i(\vec{u}) \big) \enspace ,
\end{equation*}
where $\hat{r}_i(\vec{u}) = \vec{u}^\top \hat{\vec{x}}_i$, $i \in \{1, 2, \ldots, L\}$, are the projected Dirac locations. 
\end{theorem}
\begin{proof}
    Plugging $f(\vec{x} \given \hat{\mat{X}})$ in \Eq{Eq_Dirac} into \Eq{Eq_Radon} gives
    \begin{equation*}
        f_r(r \given \hat{\vec{r}}, \vec{u}) 
        = \int_{\NewR^N} f(\vec{t} \given \hat{\mat{X}}) 
        \, \delta(r - \vec{u}^\top \vec{t}) \, \d \vec{t}
        = \sum_{i=1}^L w_i \int_{\NewR^N} \delta(\vec{t} - \hat{\vec{x}}_i)
        \, \delta(r - \vec{u}^\top \vec{t}) \, \d \vec{t} \enspace .
    \end{equation*}
    With the sifting property of the Dirac delta function $\delta(.)$, we directly obtain
    \begin{equation*}
        f_r(r \given \hat{\vec{r}}, \vec{u}) 
        = \sum_{i=1}^L w_i \, \delta(r - \vec{u}^\top \hat{\vec{x}}_i) \enspace ,
    \end{equation*}
    which is the desired result.
\end{proof}
We collect the locations in a vector 
\begin{equation*}
    \hat{\vec{r}}(\vec{u}) = \begin{bmatrix} \hat{r}_1(\vec{u}), \hat{r}_2(\vec{u}), 
    \ldots, \hat{r}_L(\vec{u}) \end{bmatrix}^\top \enspace .
\end{equation*}
The vector $\hat{\vec{r}}(\vec{u})$ is calculated as
\begin{equation*}
    \hat{\vec{r}}(\vec{u}) 
    = \big[ 
        \underbrace{\vec{u}^\top \hat{\vec{x}}_1}_{\hat{r}_1(\vec{u})}, \,
        \underbrace{\vec{u}^\top \hat{\vec{x}}_2}_{\hat{r}_2(\vec{u})}, \,
        \ldots, \,
        \underbrace{\vec{u}^\top \hat{\vec{x}}_L}_{\hat{r}_L(\vec{u})} 
    \big]^\top
    = \big( \vec{u}^\top 
    \underbrace{\big[ \hat{\vec{x}}_1, \, \hat{\vec{x}}_2, \, \ldots, \,
    \hat{\vec{x}}_L \big]}_{\hat{\mat{X}}} \big)^\top
    = \big( \vec{u}^\top \hat{\mat{X}} \big)^\top
    = \hat{\mat{X}}^\top \vec{u} \enspace .
\end{equation*}
%

%
%
\paragraph*{Gaussian Densities}
For $N$-dimensional Gaussian densities $f(\vec{x})$ given by
\begin{equation*}
    f(\vec{x}) = \frac{1}{\sqrt{(2 \pi)^N | \mat{C}_x | }}
    \exp\left( - \frac{1}{2} (\vec{x} - \hat{\vec{x}})^\top 
    \mat{C}_x^{-1} (\vec{x} - \hat{\vec{x}}) \right)
\end{equation*}
with mean vector $\hat{\vec{x}}$ and covariance matrix $\mat{C}_x$, the density $f_r(r \given \vec{u})$ with $f_r \colon \NewR \rightarrow \NewR_{\ge 0}$ resulting from the projection is also Gaussian according to
\begin{equation*}
    f_r(r \given \vec{u}) = \frac{1}{\sqrt{2 \pi} \sigma_r(\vec{u})}
    \exp\left( - \frac{1}{2} \frac{\big( r - \hat{r}(\vec{u}) \big)^2}{\sigma_r^2(\vec{u})} \right)
    \enspace .
\end{equation*}
Its mean $\hat{r}(\vec{u}) \in \NewR$ can simply be calculated by taking the expected value in \Eq{Eq_ProjectRV}, which gives
\begin{equation*}
    \hat{r}(\vec{u}) = \E\{ \rv{r}(\vec{u}) \} = \E\left\{ \vec{u}^\top \rv{\vec{x}} \right\} 
    = \vec{u}^\top \hat{\vec{x}} \enspace.
\end{equation*}
Its standard deviation $\sigma_r(\vec{u}) \in \NewR_{> 0}$  is given by
\begin{equation*}
    \sigma_r^2(\vec{u}) 
    = \E\left\{ \big( \rv{r}(\vec{u}) - \hat{r}(\vec{u}) \big)^2 \right\} 
    = \E\left\{ \left( \vec{u}^\top \rv{x} - \vec{u}^\top \hat{\vec{x}} \right)^2 \right\} 
    = \E\left\{ \vec{u}^\top ( \rv{x} - \hat{\vec{x}} ) 
        ( \rv{x} - \hat{\vec{x}} )^\top \vec{u} \right\} 
    = \vec{u}^\top \mat{C}_x \, \vec{u} \enspace .
\end{equation*}
%
%
\paragraph*{Gaussian Mixture Densities}
For $N$-dimensional Gaussian mixture densities $f(\vec{x})$ with $M$ components of the form 
\begin{equation*}
    f(\vec{x}) = \sum_{i=1}^M w_i \frac{1}{\sqrt{(2 \pi)^N | \mat{C}_{x,i} | }}
    \exp\left( - \frac{1}{2} (\vec{x} - \hat{\vec{x}}_i)^\top 
    \mat{C}_{x,i}^{-1} (\vec{x} - \hat{\vec{x}}_i)\right) \enspace ,
\end{equation*}
the density $f_r(r \given \vec{u})$ is also a Gaussian mixture, albeit a one-dimensional one. Due to the linearity of the projection operator, it is given by
\begin{equation*}
    f_r(r \given \vec{u}) = \sum_{i=1}^M w_i \frac{1}{\sqrt{2 \pi} \sigma_{r,i}(\vec{u})}
    \exp\left( - \frac{1}{2} \frac{\big( r 
    - \hat{r}_i(\vec{u}) \big)^2}{\sigma_{r,i}^2(\vec{u})} \right)
\end{equation*}
with
\begin{equation*}
    \hat{r}_i(\vec{u}) = \vec{u}^\top \hat{\vec{x}}_i
\end{equation*}
and
\begin{equation*}
    \sigma_{r,i}(\vec{u}) = \sqrt{ \vec{u}^\top \mat{C}_{x,i} \, \vec{u} }
\end{equation*}
for $i \in \{1, 2, \ldots, L\}$.

\section{One-dimensional Dirac Mixture Approximation} \label{Sec_1D}
	
	%
%
In this section, we consider the Dirac mixture approximation on the one-dimensional projections, i.e., we focus on the 1D case for one single projection vector $\vec{u}$. The one-dimensional approximation then serves as the basis for the $N$-dimensional approximation in the next section.

\subsection{Distance Measure} \label{SubSec_DM_1D}
%
%
The approximating density is given by the 1D Dirac mixture
\begin{equation}
    f(x \given \hat{\vec{x}}) = \sum_{i=1}^L w_i \, \delta(x - \hat{x}_i)
\label{Eq_Dirac1D}
\end{equation}
with given weights $w_i \in \NewR_{> 0}$ that sum up to one and locations $\hat{x}_i \in \NewR$. The desired locations are collected in a vector 
\begin{equation*}
    \hat{\vec{x}} = \begin{bmatrix} \hat{x}_1, \hat{x}_2, \ldots, \hat{x}_L 
    \end{bmatrix}^\top \enspace .
\end{equation*}
%
%
For rating the goodness of fit of $f(x \given \hat{\vec{x}})$ to the given continuous density $\tilde{f}(x)$ we compare their cumulative distribution functions rather than their densities.
%
%
For the Dirac mixture approximation $f(x \given \hat{\vec{x}})$ in \Eq{Eq_Dirac1D}, the cumulative distribution function is given by
\begin{equation*}
    F(x \given \hat{\vec{x}}) = \sum_{i=1}^L w_i \, \HeaviSide(x - \hat{x}_i) \enspace ,
\end{equation*}
where $\HeaviSide(.)$ is the Heaviside function given by
\begin{equation*}
    \HeaviSide(x) = \int_{-\infty}^x \delta(t) \, \d t 
    = \begin{cases}
        0 & x < 0 \\[2mm]
        \frac{1}{2} & x = 0 \\[2mm]
        1 & x > 0
    \end{cases}
\end{equation*}
with
\begin{equation*}
    \frac{\d}{\d x} \HeaviSide(x) = \delta(t) \enspace .
\end{equation*}
%
%
As a distance measure we use the integral squared distance between the cumulative distribution function $\tilde{F}(x)$ of the given distribution and the cumulative distribution function $F(x \given \hat{\vec{x}})$ of its Dirac mixture approximation
\begin{equation*}
    D(\hat{\vec{x}}) = \int_\NewR \left[ \tilde{F}(t) 
    - F(t \given \hat{\vec{x}}) \right]^2 \d t 
    \enspace .
\end{equation*}
The integral squared distance $D(\hat{\vec{x}})$ depends on the vector of desired locations $\hat{\vec{x}}$.

\subsection{Dirac Mixture Approximation by Minimizing the Distance Measure} \label{SubSec_Newton_1D}
%
%
The gradient of the distance measure $D(\hat{\vec{x}})$ is given by
\begin{equation*}
    \vec{G}(\hat{\vec{x}}) = \nabla D(\hat{\vec{x}}) 
    = \frac{\partial D(\hat{\vec{x}})}{\partial \hat{\vec{x}}} 
    = \begin{bmatrix} 
        \frac{\partial D(\hat{\vec{x}})}{\partial \hat{x}_1}, &
        \frac{\partial D(\hat{\vec{x}})}{\partial \hat{x}_2}, &
        \ldots, &
        \frac{\partial D(\hat{\vec{x}})}{\partial \hat{x}_L}
    \end{bmatrix}^\top
\end{equation*}
with
\begin{equation*}
    \frac{\partial D(\hat{\vec{x}})}{\partial \hat{x}_i}
    = 2 w_i \int_\NewR \left[ \tilde{F}(t) - F(t \given \hat{\vec{x}}) \right] 
    \delta(t - \hat{x}_i) \, \d t
    = 2 w_i \left[ \tilde{F}(\hat{x}_i) - F(\hat{x}_i \given \hat{\vec{x}}) \right]
    \enspace .
\end{equation*}
%
%
The empirical cumulative distribution function $F(. \given .)$ evaluated at location $\hat{x}_i$ is given by
\begin{equation}
    F(\hat{x}_i \given \hat{\vec{x}}) = \sum_{j=1}^L w_j \, \HeaviSide(\hat{x}_i - \hat{x}_j) \enspace .
\label{Eq_eCDF_i}
\end{equation}
\begin{remark}
    It is interesting to note that $F(\hat{x}_i \given \hat{\vec{x}})$ does not depend upon $\hat{x}_i$. As $\HeaviSide(\hat{x}_i - \hat{x}_j) = 1/2$ for $i=j$, we have
    \begin{equation*}
    F(\hat{x}_i \given \hat{\vec{x}}) = \frac{w_i}{2} + \sum_{j=1 \atop j \neq i}^L w_j \, \HeaviSide(\hat{x}_i - \hat{x}_j) \enspace .
    \end{equation*}
    For $i \neq j$, we have $\hat{x}_i \neq \hat{x}_j$ and $\HeaviSide(\hat{x}_i - \hat{x}_j)$ is either $0$ or $1$ and does not depend upon $\hat{x}_i$.
\end{remark}
%
%
The Hesse matrix is defined as
\def\Dx{D(\hat{\vec{x}})}
\begin{equation*}
    \mat{H}(\hat{\vec{x}})
    = \nabla \Dx \, \nabla^\top 
    = \frac{\partial^2 \Dx}{\partial \hat{\vec{x}} \, \partial \hat{\vec{x}}^\top}
    = 
    \begin{bmatrix} 
        \frac{\partial^2 \Dx}{\partial \hat{x}_1^2} &
        \frac{\partial^2 \Dx}{\partial \hat{x}_1 \, \partial \hat{x}_2} &
        \ldots &
        \frac{\partial^2 \Dx}{\partial \hat{x}_1 \, \partial \hat{x}_L} \\[4mm]
        \frac{\partial^2 \Dx}{\partial \hat{x}_2 \, \partial \hat{x}_1} &
        \frac{\partial^2 \Dx}{\partial \hat{x}_2^2} &
        \ldots &
        \frac{\partial^2 \Dx}{\partial \hat{x}_2 \, \partial \hat{x}_L} \\[4mm]
        \vdots & \vdots & \vdots \\[4mm]
        \frac{\partial^2 \Dx}{\partial \hat{x}_L \, \partial \hat{x}_1} &
        \frac{\partial^2 \Dx}{\partial \hat{x}_L \, \partial \hat{x}_2} &
        \ldots &
        \frac{\partial^2 \Dx}{\partial \hat{x}_L^2}
    \end{bmatrix} \enspace .
\end{equation*}
With
\begin{equation*}
    \frac{\partial^2 \Dx}{\partial \hat{x}_i \, \partial \hat{x}_j} = 0
\end{equation*}
for $i \neq j$, we have
\begin{equation*}
    \mat{H}(\hat{\vec{x}}) = \diag\left( 
        \begin{bmatrix} 
            \frac{\partial^2 \Dx}{\partial \hat{x}_1^2},
            \frac{\partial^2 \Dx}{\partial \hat{x}_2^2},
            \ldots ,
            \frac{\partial^2 \Dx}{\partial \hat{x}_L^2}
        \end{bmatrix} 
        \right) \enspace ,
\end{equation*}
where
\begin{equation*}
    \frac{\partial^2 D(\hat{\vec{x}})}{\partial \hat{x}_i^2} 
    = 2 \, w_i \, \tilde{f}(\hat{x}_i)
\end{equation*}
for $i \in \{1, 2, \ldots, L\}$.
%
%
The minimum of $D(\hat{\vec{x}})$ is obtained iteratively using Newton's method
\begin{equation*}
    \DeltaX{x}(\hat{\vec{x}}) = - \mat{H}^{-1}(\hat{\vec{x}}) \, \vec{G}(\hat{\vec{x}}) \enspace ,
\end{equation*}
starting with an initial location vector.

%
%
\paragraph*{Special Case -- Sorted Locations}
When the location vector $\hat{\vec{x}}$ is sorted, i.e., $\hat{x}_1 < \hat{x}_2 < \ldots < \hat{x}_L$, the expression $H(\hat{x}_i - \hat{x}_j)$ in \Eq{Eq_eCDF_i} gives
\begin{equation*}
    H(\hat{x}_i - \hat{x}_j)
    = \begin{cases}
        0 & i < j \\[2mm]
        \frac{1}{2} & i = j \\[2mm]
        1 & i > j
    \end{cases} \enspace .
\end{equation*}
As a result, $F(\hat{x}_i \given \hat{\vec{x}})$ can be written as
\begin{equation*}
    F(\hat{x}_i \given \hat{\vec{x}}) = \frac{w_i}{2} + \sum_{j=1}^{i-1} w_j \enspace .
\end{equation*}
%
%
%
\paragraph*{Special Case -- Sorted Locations \& Equal Weights}
For equally weighted samples, i.e., $w_i = 1/L$, the expression can be simplified to
\begin{equation*}
    F(\hat{x}_i \given \hat{\vec{x}}) = \frac{1}{2 L} + (i-1) \frac{1}{L} = \frac{2 i -1}{2 L} \enspace .
\end{equation*}

\newcommand{\QuadFigure}[3]{%
\begin{figure*}[t]
\includegraphics[width=82.5mm]{Figures/#11.pdf}
\includegraphics[width=82.5mm]{Figures/#12.pdf} \\[-2mm]
\rule{\textwidth}{1pt}\\[2mm]
\includegraphics[width=82.5mm]{Figures/#13.pdf}
\includegraphics[width=82.5mm]{Figures/#14.pdf}
\caption{#2}
\label{#3}
\end{figure*}
}%

\QuadFigure{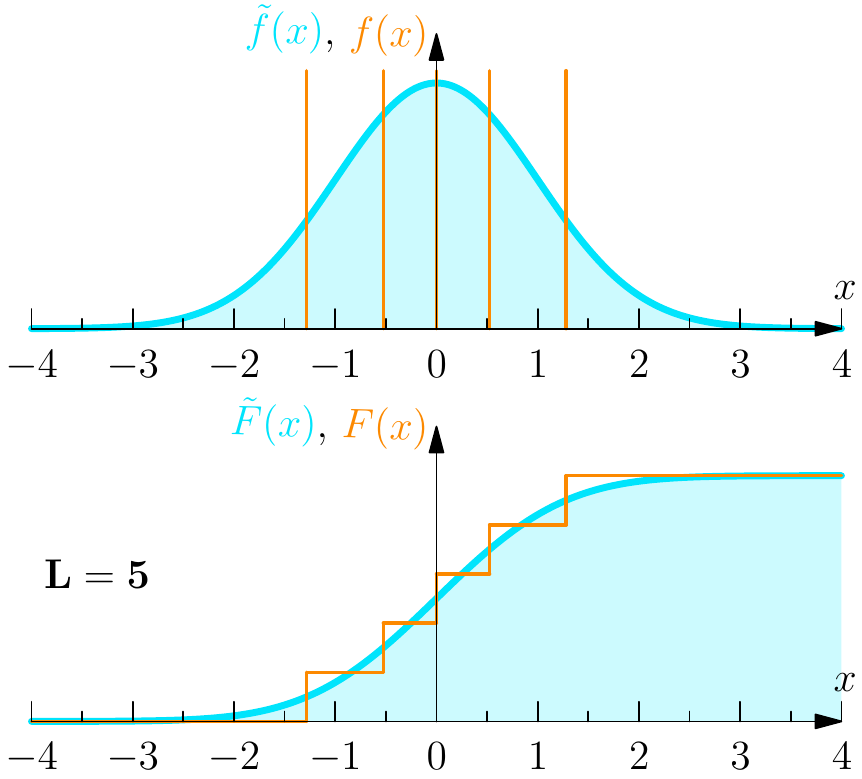}{Dirac mixture approximation with equally weighted samples of univariate standard normal distribution with different numbers of samples $L$.}{Fig_DMA_SND}

\begin{example}[Univariate standard normal distribution]
%
%
We first consider the approximation of a one-dimensional standard normal distribution $\tilde{f}(x)$ given by
\begin{equation*}
    \tilde{f}(x) = \frac{1}{\sqrt{2 \pi}} \exp\left( - \frac{1}{2} x^2 \right) \enspace .
\end{equation*}
%
%
The cumulative distribution function of the standard normal distribution $\tilde{f}(x)$ is given by
\begin{equation*}
    \tilde{F}(x) = \frac{1}{2} \left( 1 + \erf\left( \frac{x}{\sqrt{2}} \right) \right) \enspace ,
\end{equation*}
with $\erf(.)$ the error function\footnote{The error function $\erf(.)$ is defined as $\erf(x) = \displaystyle\frac{2}{\sqrt{\pi}} \int_0^x \exp(-t^2) \, \d t$.}.
%
%
In this case, a closed-form solution of $\vec{G}(\hat{\vec{x}}) = \vec{0}$ is available, which gives locations
\begin{equation*}
    \hat{x}_i = \sqrt{2} \, \erfinv\left( \frac{2 i - 1 - L}{L} \right)
    \text{ for } i = 1, \ldots, L \enspace .
\end{equation*}
The function $\erfinv(.)$ is the inverse of the error function. Dirac mixture approximations of the standard normal distribution with different numbers of components $L$ are given in \Fig{Fig_DMA_SND}.
\end{example}

\QuadFigure{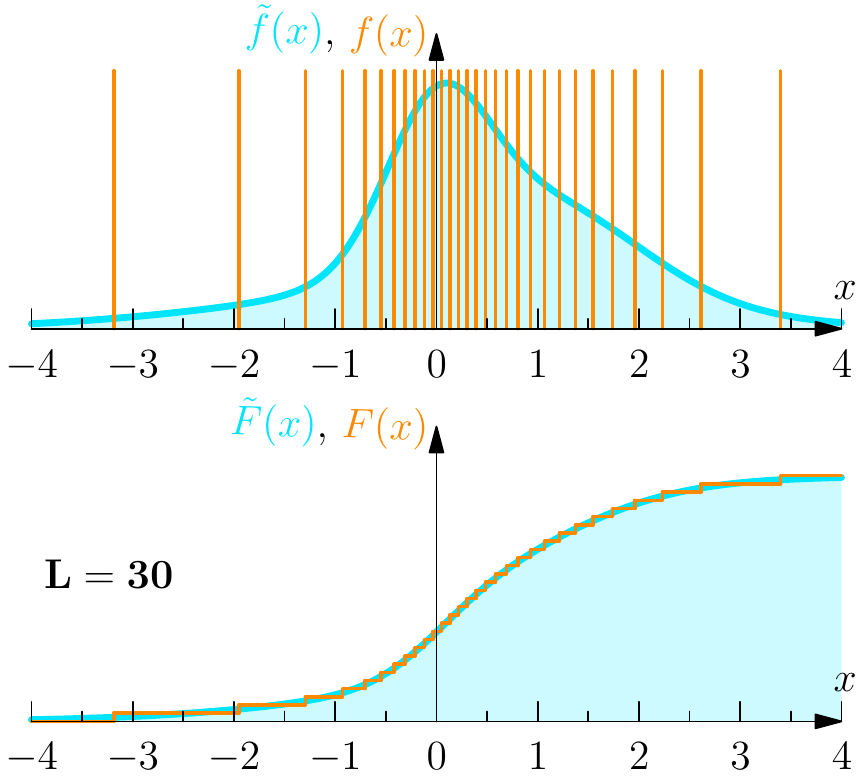}{Dirac mixture approximation with equally weighted samples of some Gaussian mixtures with various components and parameters.}{Fig_DMA_1D_GM}

\begin{example}[Univariate Gaussian mixtures]
For Gaussian mixtures, the distance measure is minimized as described above. Dirac mixture approximations with $L=20$ of various Gaussian mixtures with three components with different parameters are given in \Fig{Fig_DMA_1D_GM}.
\end{example}

\section{Multivariate Dirac Mixture Approximation} \label{Sec_ND}
	
	%
%
\noindent
We are now given the three main ingredients:
\begin{enumerate}
\item We can represent a multivariate probability density function by its one-dimensional projections, see \Sec{Sec_Radon}. 
\item We have a mechanism for calculating a distance measure between a one-dimensional continuous density and its Dirac mixture approximation, see \SubSec{SubSec_DM_1D}. 
\item We can minimize this distance measure with a Newton-like method, see \SubSec{SubSec_Newton_1D}.
\end{enumerate}

%
%
Based on these three ingredients, we will now assemble a multivariate distance measure between two continuous and/or discrete probability density functions. This entails the following seven steps.

%
%
\paragraph*{Step 1 [One-dimensional Projections via Radon Transform]}
The given density $\tilde{f}(\vec{x})$ and its Dirac mixture approximation $f(\vec{x} \given \hat{\mat{X}})$ in \Eq{Eq_Dirac} are represented by their Radon transforms $\tilde{f}(r \given \vec{u})$ and $f(r \given \hat{\vec{r}}, \vec{u})$, i.e., by their 1D projections onto unit vectors $\vec{u} \in \mathbb{S}^{N-1}$.

%
%
\paragraph*{Step 2 [One-dimensional Cumulative Distributions]}
Based on the Radon transform $\tilde{f}(r \given \vec{u})$, we calculate the one-dimensional cumulative distributions of the projected densities as
\begin{equation*}
    \tilde{F}(r \given \vec{u}) = \int_{\infty}^{r} \tilde{f}(t, \vec{u}) \, \d t
\end{equation*}
and similarly for $F(r \given \vec{u})$.

%
%
For a Dirac mixture approximation, the cumulative distribution function of its Radon transform is given by
\begin{equation*}
    F(r \given \hat{\vec{r}}, \vec{u}) = \sum_{i=1}^L w_i \, \text{H} \big( r - \hat{r}_i(\vec{u}) \big) \enspace .
\end{equation*}
%

%
%
\paragraph*{Step 3 [One-dimensional Distance]}
For comparing the one-dimensional projections, we compare their cumulative distributions $\tilde{F}(r \given \vec{u})$ and $F(r \given \hat{\vec{r}}, \vec{u})$ for all $\vec{u} \in \mathbb{S}^{N-1}$. As a distance measure we use the integral squared distance
\begin{equation}
    D_1(\hat{\vec{r}}, \vec{u}) 
    = \int_\NewR \left[ \tilde{F}(r \given \vec{u}) - F(r \given \hat{\vec{r}}, \vec{u}) \right]^2 \, \d r \enspace .
\label{Eq_D1}
\end{equation}
This gives us the distance between the projected densities in the direction of the unit vector $\vec{u}$ for all $\vec{u}$.

%
%
\paragraph*{Step 4 [One-dimensional Newton Step]}
The Newton step from \SubSec{SubSec_Newton_1D} can now be written as 
\begin{equation*}
    \DeltaX{r}(\hat{\vec{r}}, \vec{u}) = - \mat{H}^{-1}(\hat{\vec{r}}, \vec{u}) \, \vec{G}(\hat{\vec{r}}, \vec{u})
\end{equation*}
with
\begin{equation*}
    \vec{G}(\hat{\vec{r}}, \vec{u})
    = \begin{bmatrix} 
        \displaystyle
        \frac{\partial D_1(\hat{\vec{r}}, \vec{u})}{\partial \hat{r}_1}, &
        \displaystyle
        \frac{\partial D_1(\hat{\vec{r}}, \vec{u})}{\partial \hat{r}_2}, &
        \ldots, &
        \displaystyle
        \frac{\partial D_1(\hat{\vec{r}}, \vec{u})}{\partial \hat{r}_L}
    \end{bmatrix}^\top
\end{equation*}
and
\begin{equation*}
    \frac{\partial D_1(\hat{\vec{r}}, \vec{u})}{\partial \hat{r}_i}
    = 2 w_i \left[ \tilde{F}(\hat{r}_i \given \vec{u}) - F(\hat{r}_i \given \hat{\vec{r}}, \vec{u}) \right]
\end{equation*}
for $i \in \{1, 2, \ldots, L\}$. The Hessian $\mat{H}(\hat{\vec{r}}, \vec{u})$ is given by
\begin{equation*}
    \mat{H}(\hat{\vec{r}}, \vec{u}) 
    = 2 \diag\left( 
        \begin{bmatrix} 
            w_1 \, \tilde{f}(\hat{r}_1 \given \vec{u}) , &
            w_2 \, \tilde{f}(\hat{r}_2 \given \vec{u}) , &
            \ldots , &
            w_L \, \tilde{f}(\hat{r}_L \given \vec{u})
        \end{bmatrix} 
    \right) \enspace .
\end{equation*}
The resulting Newton step is
\begin{equation*}
    \DeltaX{r}(\hat{\vec{r}}, \vec{u}) = -
        \begin{bmatrix} 
            \hat{\Delta}_{r, 1}(\hat{r}_1, \vec{u}) , &
            \hat{\Delta}_{r, 2}(\hat{r}_2, \vec{u}) , &
            \ldots , &
            \hat{\Delta}_{r, L}(\hat{r}_L, \vec{u})
        \end{bmatrix}^\top
\end{equation*}
with
\begin{equation*}
    \hat{\Delta}_{r, i}(\hat{r}_i, \vec{u}) =
        \displaystyle
        \frac{\tilde{F}(\hat{r}_i \given \vec{u}) 
        - F(\hat{r}_i \given \hat{\vec{r}}, \vec{u})}{\tilde{f}(\hat{r}_i \given \vec{u})}
\end{equation*}
for $i \in \{1, 2, \ldots, L\}$.
%
%
\paragraph*{Step 5 [Backprojection to ND space]}
For a specific projection vector $\vec{u}$, we obtain a Newton update $\DeltaX{x,i}(\vec{u})$. By means of a backprojection into the original $N$-dimensional space, this update can be used to modify the original Dirac locations in the direction \emph{along} the vector $\vec{u}$. For every location vector $\hat{\vec{x}}_i$, its modification $\DeltaX{x,i}(\vec{u})$ is given by
\begin{equation}
    \DeltaX{x,i}(\vec{u}) = \hat{\Delta}_{r, i}(\hat{r}_i, \vec{u}) \, \vec{u}
\label{Eq_backPropo_single}
\end{equation}
\def\DeltaXmat{\hat{\boldsymbol{\Delta}}_{x}}
for $i \in \{1, 2, \ldots, L\}$. Collecting $\DeltaX{x,i}(\vec{u})$ as column vectors in a matrix $\DeltaXmat$ gives
\begin{equation*}
\DeltaXmat(\vec{u}) = \begin{bmatrix} \DeltaX{x,1}(\vec{u}) , & \DeltaX{x,2}(\vec{u}), &
\ldots, & \DeltaX{x,L}(\vec{u})\end{bmatrix} \enspace .
\end{equation*}
Based on \Eq{Eq_backPropo_single}, it is calculated as
\begin{equation*}
    \DeltaXmat(\vec{u}) = \vec{u} \, \DeltaX{r}^\top \enspace .
\end{equation*}
%

%
%
\paragraph*{Step 6 [Assemble Multivariate Distance]}
The individual 1D distances $D_1(\vec{r}, \vec{u})$ can be assembled to form a multivariate distance measure. This is performed by averaging over all 1D distances depending on the unit vector $\vec{u}$
\begin{equation*}
    D_N(\hat{\mat{X}}) = \frac{1}{A_N} \int_{\NewS^{N-1}} D_1(\hat{\vec{r}}, \vec{u})  \, \d \vec{u} \enspace .
\end{equation*}
Here, $A_N$ is the surface area of the (hyper)sphere $\NewS^{N-1}$ embedded in $\NewR^N$. $A_N$ is given by
\begin{equation*}
    A_N = \frac{2 \, \pi^{N/2}}{\Gamma(N/2)} 
\end{equation*}
with $\Gamma(.)$ the gamma function and can be calculated recursively as
\begin{equation*}
    A_0 = 0, \; A_1 = 2, \; A_2 = 2 \pi, \; 
    A_N = \frac{2 \pi}{N-2} \, A_{N-2} \text{ for } N>2 \enspace .
\end{equation*}
Plugging in $D_1(\vec{r}, \vec{u})$ from \Eq{Eq_D1} gives
\begin{equation*}
    D_N(\hat{\mat{X}}) = \frac{1}{A_N} \int_{\mathbb{S}^{N-1}} 
    \int_\NewR \left[ \tilde{F}(r \given \vec{u}) 
    - F(r \given \hat{\vec{r}}, \vec{u}) \right]^2 \, \d r \, \d \vec{u}
\end{equation*}
with $r = \vec{u}^\top \, \vec{x}$.

%
%
\paragraph*{Step 7 [Perform Full Newton Update]}
A full Newton update can now be performed by averaging over all partial updates along projection vectors $\vec{u}$
\begin{equation*}
    \DeltaX{x,i} = \frac{1}{A_N} \int_{\mathbb{S}^{N-1}} 
    \DeltaX{x,i}(\vec{u}) \, \d \vec{u}
    \enspace .
\end{equation*}
For the matrix $\DeltaXmat$, we obtain
\begin{equation*}
    \DeltaXmat = \frac{1}{A_N} \int_{\mathbb{S}^{N-1}} \DeltaXmat(\vec{u}) \, \d \vec{u}
    \enspace .
\end{equation*}
%
%
%
In a practical implementation, the space $\mathbb{S}^{N-1}$ containing the unit vectors $\vec{u}$ has to be discretized.
%
%
Two options are available for performing the discretization: (1) Deterministic discretization, e.g., by calculating a grid or (2) random discretization by drawing uniform samples from the hypersphere.
%
%
In both cases, we consider $K$ samples $\hat{\vec{u}}_k$ and the integration reduces to a summation
\begin{equation*}
    \DeltaX{x,i} \approx \frac{1}{K} \sum_{k=1}^K \DeltaX{x,i}(\hat{\vec{u}}_k)
\end{equation*}
for $i \in \{1, 2, \ldots, L\}$ or
\begin{equation*}
    \DeltaXmat \approx \frac{1}{K} \sum_{k=1}^K \DeltaXmat(\hat{\vec{u}}_k) \enspace .
\end{equation*}
%
%
Given initial locations for the location of the Dirac components, full Newton updates are performed as
\begin{equation*}
    \hat{\mat{X}} \pluseq \DeltaXmat
\end{equation*}
until the maximum change over all location vectors
\begin{equation*}
    \max_i \left\| \DeltaX{x,i} \right\|
\end{equation*}
falls below a given threshold.

\begin{figure*}[t]
\hspace*{5mm}
\parbox{75mm}{
{\bfseries Input:}
\begin{itemize}
\setlength{\itemsep}{3pt}
\setlength{\parskip}{0pt}
\setlength{\parsep}{0pt}
\item Number of dimensions $N$
\item Number of Dirac components $L$
\item Number of discretizations $K$
\item Step size $s$
\item Initial location matrix $\hat{\mat{X}} \in \NewR^{N \times L}$
\end{itemize}
{\bfseries Output:}
\begin{itemize}
\setlength{\itemsep}{3pt}
\setlength{\parskip}{0pt}
\setlength{\parsep}{0pt}
\item Optimal location matrix $\hat{\mat{X}}$
\end{itemize}
}
\parbox{100mm}{
\begin{algorithm}[H]
\DontPrintSemicolon
\Repeat{Changes $\DeltaXmat$ small enough}{
    \tcp{Initialize matrix of location changes $\DeltaXmat$}        
    $\DeltaXmat = \mat{0}$\;
    \For{$i=1$ \KwTo $K$}{
        \tcp{Random unit vector $\vec{u} \in \NewS^{N-1}$}
        $\vec{u} = \mathrm{randn}(D)$\;        
        $\vec{u} \diveq \sqrt{\vec{u}^\top \vec{u}}$\;    
        \tcp{Projection of samples onto vector $\vec{u}$}
        $\hat{\vec{r}} = \hat{\mat{X}}^\top \vec{u}$\; 
        \tcp{Newton update along unit vector $\vec{u}$}
        $\DeltaX{r}
            = - \displaystyle\frac{\tilde{F}(\hat{\vec{r}} \given \vec{u}) 
            - F(\hat{\vec{r}} \given \hat{\vec{r}}, 
            \vec{u})}{\tilde{f}(\hat{\vec{r}} \given \vec{u})}$\;
        \tcp{Backprojection of location changes}
        $\DeltaXmat \pluseq \mathop{s} \, \vec{u} \, \DeltaX{r}^\top/K$\;
    }
    \tcp{Update location matrix $\hat{\mat{X}}$}
    $\hat{\mat{X}} \pluseq \DeltaXmat$\;
}
\end{algorithm}}
\caption{Complete algorithm for multivariate Dirac mixture approximation.}
\label{Fig_Algorithm}
\end{figure*}

The complete algorithm for multivariate Dirac mixture approximation is shown in \Fig{Fig_Algorithm}.

\section{Numerical Results} \label{Sec_Numerical}
	
	%
%
In this section, we look at some example approximations.
%
%
We focus on 2D examples for visualization purposes.
%
%
To keep the setup simple, the true continuous densities are selected as Gaussian mixture densities with two components and varying parameters.

\begin{example}[2D Gaussian mixture with two components, varying means]
We start with a Gaussian mixture with two components. Weights are equal, i.e., $w_1 = 0.5$, $w_2 = 0.5$, and both covariance matrices are identity matrices, i.e., $\mat{C}_{x,1} = \mat{I}$, $\mat{C}_{x,2} = \mat{I}$, see \Fig{Fig_DMA_50_Gauss_2D_means}. The means are selected as follows:
\begin{enumerate}
\item $\hat{\vec{x}}_1 = [\phantom{-}0.0,0.0]^\top$, $\hat{\vec{x}}_2 = [0.0,0.0]^\top$ (this corresponds to a standard normal distribution),
\item $\hat{\vec{x}}_1 = [-0.7,0.0]^\top$, $\hat{\vec{x}}_2 = [0.7,0.0]^\top$,
\item $\hat{\vec{x}}_1 = [-1.4,0.0]^\top$, $\hat{\vec{x}}_2 = [1.4,0.0]^\top$,
\item $\hat{\vec{x}}_1 = [-2.1,0.0]^\top$, $\hat{\vec{x}}_2 = [2.1,0.0]^\top$.
\end{enumerate}
The results are shown in \Fig{Fig_DMA_50_Gauss_2D_means} for an approximation with $L=50$ Dirac components. The results for $L=100$ Dirac components are shown in \Fig{Fig_DMA_100_Gauss_2D_means}.
\end{example}

\begin{example}[2D Gaussian mixture with two components, varying covariance matrices]
We start with a standard normal distribution represented by a Gaussian mixture with equal components, i.e., $w_1 = 0.5$, $w_2 = 0.5$, $\hat{\vec{x}}_1 = [0,0]^\top$, $\hat{\vec{x}}_2 = [0,0]^\top$, see \Fig{Fig_DMA_50_Gauss_2D_covs}. The covariance matrices are selected as follows:
\begin{enumerate}
\item $\mat{C}_{x,1} = \begin{bmatrix} 3.0 & 1.5 \\ 1.5 & 3.0 \end{bmatrix}$, 
$\mat{C}_{x,2} = \begin{bmatrix} 3.0 & -1.5 \\ -1.5 & 3.0 \end{bmatrix}$,
\item $\mat{C}_{x,1} = \begin{bmatrix} 3.0 & 2.0 \\ 2.0 & 3.0 \end{bmatrix}$, 
$\mat{C}_{x,2} = \begin{bmatrix} \phantom{-}3.0 & -2.0 \\ -2.0 & \phantom{-}3.0 \end{bmatrix}$,
\item $\mat{C}_{x,1} = \begin{bmatrix} 3.0 & 2.5 \\ 2.5 & 3.0 \end{bmatrix}$, 
$\mat{C}_{x,2} = \begin{bmatrix} \phantom{-}3.0 & -2.5 \\ -2.5 & \phantom{-}3.0 \end{bmatrix}$,
\item $\mat{C}_{x,1} = \begin{bmatrix} 3.0 & 2.8 \\ 2.8 & 3.0 \end{bmatrix}$, 
$\mat{C}_{x,2} = \begin{bmatrix} \phantom{-}3.0 & -2.8 \\ -2.8 & \phantom{-}3.0 \end{bmatrix}$.
\end{enumerate}
The results are shown in \Fig{Fig_DMA_50_Gauss_2D_covs} for an approximation with $L=50$ Dirac components. The results for $L=100$ Dirac components are shown in \Fig{Fig_DMA_100_Gauss_2D_covs}.
\end{example}

\QuadFigure{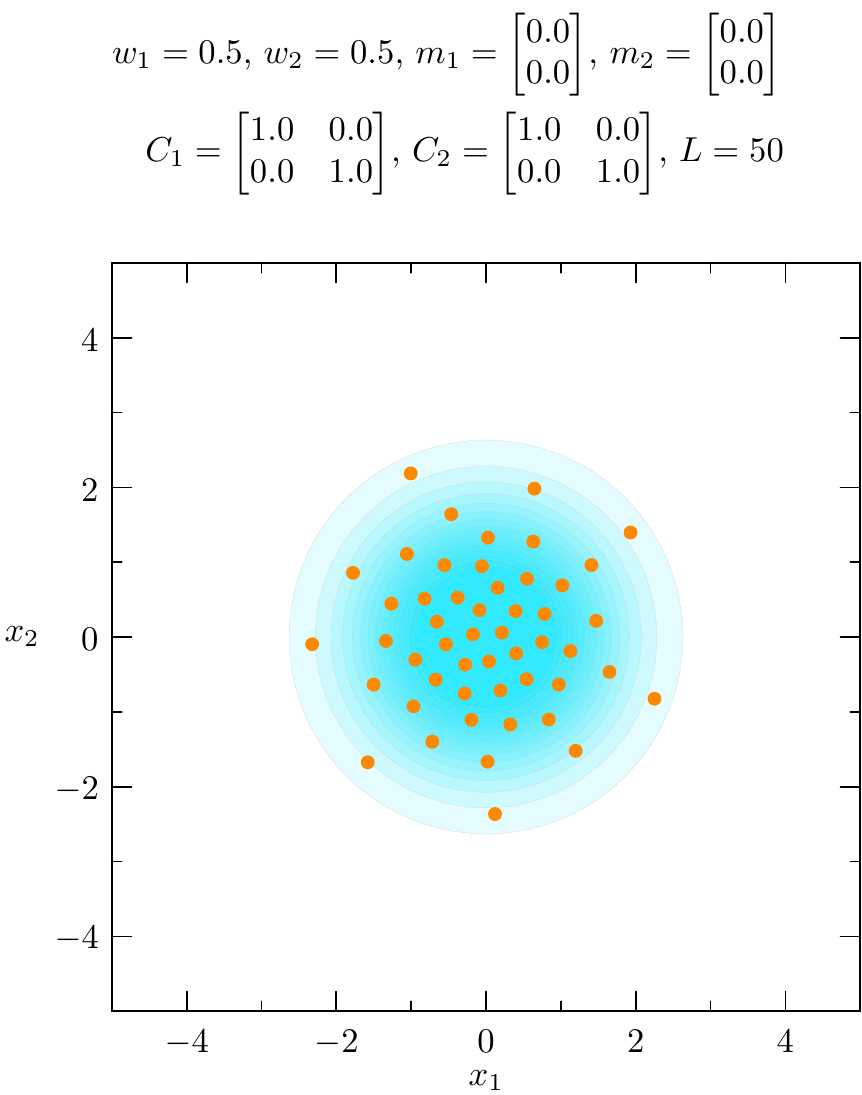}{Dirac mixture approximation with $L=50$ components of Gaussian mixture approximation with two components and varying means.}{Fig_DMA_50_Gauss_2D_means}

\QuadFigure{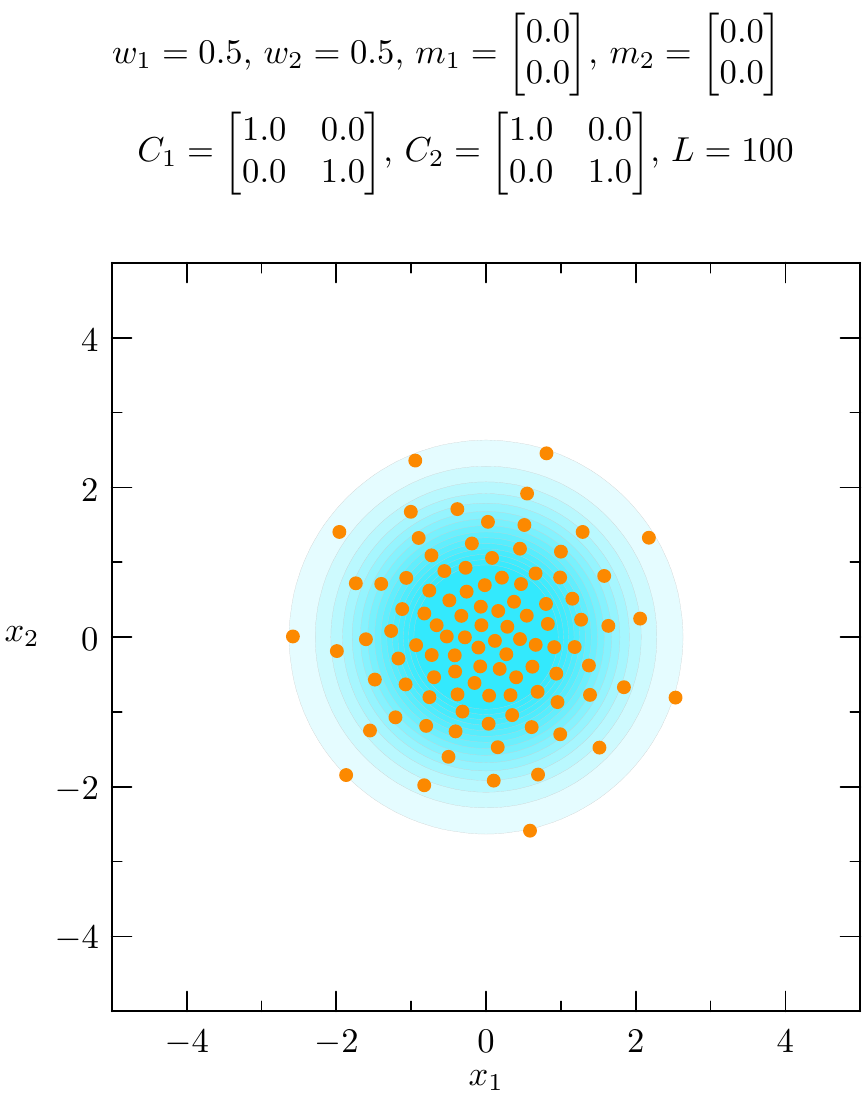}{Dirac mixture approximation with $L=100$ components of Gaussian mixture approximation with two components and varying means.}{Fig_DMA_100_Gauss_2D_means}

\QuadFigure{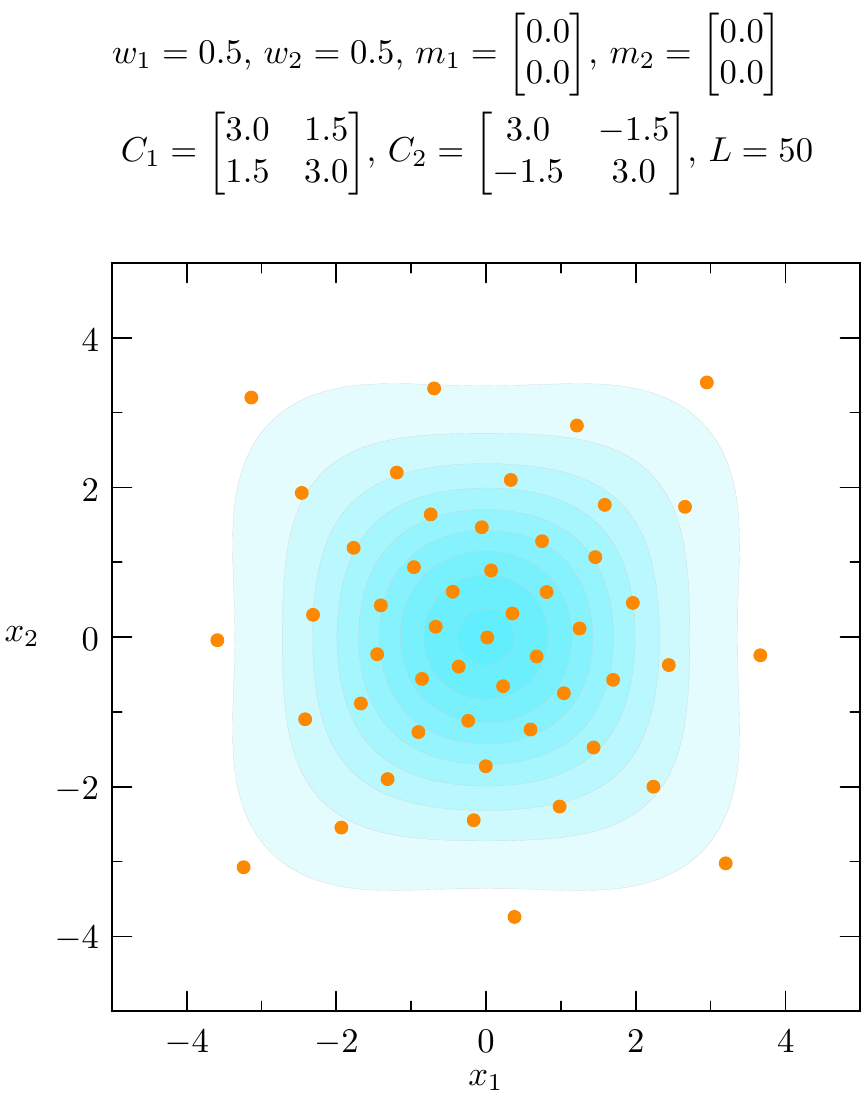}{Dirac mixture approximation with $L=50$ components of Gaussian mixture approximation with two components and varying covariance matrices.}{Fig_DMA_50_Gauss_2D_covs}

\QuadFigure{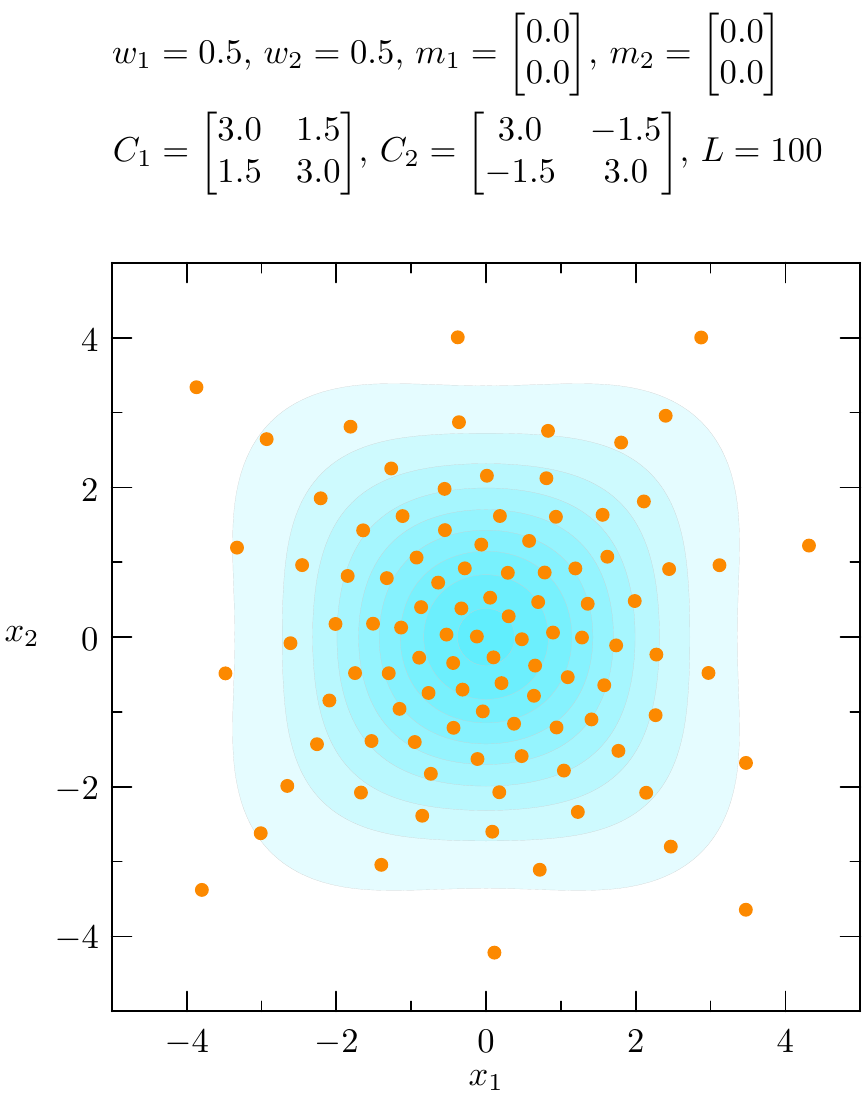}{Dirac mixture approximation with $L=100$ components of Gaussian mixture approximation with two components and varying covariance matrices.}{Fig_DMA_100_Gauss_2D_covs}

\section{Conclusions} \label{Sec_Conclude}

	%
%
This paper proposes a new integral transformation of probability density functions (PDF). The so-called Projected Cumulative Distribution (PCD) is used for characterizing both continuous and discrete random variables. It is a convenient alternative to the standard cumulative distribution function (CDF).
%
%
In multivariate settings, the CDF is non-unique, difficult to calculate, and asymmetric. This leads to various problems when used in density comparisons.
%
%
The PCD does not suffer from these problems, can be efficiently calculated, and is simple to implement.

%
%
For comparing two densities, a generalized Cram\'{e}r-von Mises distance between their two PCDs is derived.
%
%
For a given continuous density $\tilde{f}(\vec{x})$, its Dirac mixture approximation $f(\vec{x})$ is then obtained by minimizing the distance measure with respect to its Dirac locations.
%
%
Minimization is performed by combining Newton updates from all (discrete) projections.

%
%
The proposed density approximation is parallelizable. Of course, it cannot be parallelized in terms of the Dirac components as these are dependent on each other. However, the combined Newton step can be performed in parallel for each projection.

%
%
So far, we considered analytic calculation of the projections. This is only possible in special cases. For general multivariate densities, exact analytic expressions may not be possible. 
%
%
Two options are available. The first option is to approximate the given density by a more convenient density, e.g., a Gaussian mixture. Projections are then calculated for this approximate density. The second option is to perform numeric line integrals in the Radon transform. The integration only has to be performed once for a given density $\tilde{f}(\vec{x})$.


\bibliographystyle{IEEEtran}

\bibliography{%
        BibTex/Library,%
        BibTex/ISAS-Publications/ISASPublikationen,%
        BibTex/ISAS-Publications/ISASPublikationen_laufend,%
        BibTex/ISAS-Publications/ISASPreprints%
        }

\begin{thebibliography}{10}
\providecommand{\url}[1]{#1}
\csname url@samestyle\endcsname
\providecommand{\newblock}{\relax}
\providecommand{\bibinfo}[2]{#2}
\providecommand{\BIBentrySTDinterwordspacing}{\spaceskip=0pt\relax}
\providecommand{\BIBentryALTinterwordstretchfactor}{4}
\providecommand{\BIBentryALTinterwordspacing}{\spaceskip=\fontdimen2\font plus
\BIBentryALTinterwordstretchfactor\fontdimen3\font minus
  \fontdimen4\font\relax}
\providecommand{\BIBforeignlanguage}[2]{{%
\expandafter\ifx\csname l@#1\endcsname\relax
\typeout{** WARNING: IEEEtran.bst: No hyphenation pattern has been}%
\typeout{** loaded for the language `#1'. Using the pattern for}%
\typeout{** the default language instead.}%
\else
\language=\csname l@#1\endcsname
\fi
#2}}
\providecommand{\BIBdecl}{\relax}
\BIBdecl

\bibitem{lefebvreLinearRegressionKalman2005}
T.~Lefebvre, H.~Bruyninckx, and J.~De~Schutter, ``The {{Linear Regression
  Kalman Filter}},'' in \emph{Nonlinear {{Kalman Filtering}} for
  {{Force}}-{{Controlled Robot Tasks}}}, ser. Springer {{Tracts}} in {{Advanced
  Robotics}}, 2005, vol.~19.

\bibitem{julierNewMethodNonlinear2000}
S.~Julier, J.~Uhlmann, and H.~F. {Durrant-Whyte},
  ``\BIBforeignlanguage{English}{A {{New Method}} for the {{Nonlinear
  Transformation}} of {{Means}} and {{Covariances}} in {{Filters}} and
  {{Estimators}}},'' \emph{\BIBforeignlanguage{English}{IEEE Transactions on
  Automatic Control}}, vol.~45, no.~3, pp. 477--482, Mar. 2000.

\bibitem{julierScaledUnscentedTransformation2002}
S.~J. Julier, ``\BIBforeignlanguage{English}{The {{Scaled Unscented
  Transformation}}},'' in \emph{\BIBforeignlanguage{English}{Proceedings of the
  2002 {{IEEE American Control Conference}} ({{ACC}} 2002)}}, vol.~6,
  {Anchorage, Alaska, USA}, May 2002, pp. 4555-- 4559.

\bibitem{tenneHigherOrderUnscented2003}
D.~Tenne and T.~Singh, ``\BIBforeignlanguage{English}{The {{Higher Order
  Unscented Filter}}},'' in \emph{\BIBforeignlanguage{English}{Proceedings of
  the 2003 {{IEEE American Control Conference}} ({{ACC}} 2003)}}, vol.~3,
  {Denver, Colorado, USA}, Jun. 2003, pp. 2441--2446.

\bibitem{arasaratnamCubatureKalmanFilters2009}
I.~Arasaratnam and S.~Haykin, ``\BIBforeignlanguage{English}{Cubature {{Kalman
  Filters}}},'' \emph{\BIBforeignlanguage{English}{IEEE Transactions on
  Automatic Control}}, vol.~54, no.~6, pp. 1254--1269, Jun. 2009.

\bibitem{masseyKolmogorovSmirnovTestGoodness1951}
F.~J. Massey, ``\BIBforeignlanguage{en}{The {{Kolmogorov}}-{{Smirnov Test}} for
  {{Goodness}} of {{Fit}}},'' \emph{\BIBforeignlanguage{en}{Journal of the
  American Statistical Association}}, vol.~46, no. 253, pp. 68--78, Mar. 1951.

\bibitem{cramerCompositionElementaryErrors1928}
H.~Cram{\'e}r, ``\BIBforeignlanguage{en}{On the {{Composition}} of {{Elementary
  Errors}}: {{Second Paper}}: {{Statistical Applications}}},''
  \emph{\BIBforeignlanguage{en}{Scandinavian Actuarial Journal}}, vol. 1928,
  no.~1, pp. 141--180, Jan. 1928.

\bibitem{andersonAsymptoticTheoryCertain1952}
T.~W. Anderson and D.~A. Darling, ``\BIBforeignlanguage{EN}{Asymptotic
  {{Theory}} of {{Certain}} ``{{Goodness}} of {{Fit}}'' {{Criteria Based}} on
  {{Stochastic Processes}}},'' \emph{\BIBforeignlanguage{EN}{The Annals of
  Mathematical Statistics}}, vol.~23, no.~2, pp. 193--212, Jun. 1952.

\bibitem{CDC06_Schrempf-DiracMixt}
O.~C. Schrempf, D.~Brunn, and U.~D. Hanebeck, ``{D}ensity {A}pproximation
  {B}ased on {D}irac {M}ixtures with {R}egard to {N}onlinear {E}stimation and
  {F}iltering,'' in \emph{Proceedings of the 2006 IEEE Conference on Decision
  and Control (CDC 2006)}, San Diego, California, USA, Dec. 2006.

\bibitem{MFI06_Schrempf-CramerMises}
------, ``{Dirac Mixture Density Approximation Based on Minimization of the
  Weighted Cram\'{e}r-von Mises Distance},'' in \emph{Proceedings of the 2006
  IEEE International Conference on Multisensor Fusion and Integration for
  Intelligent Systems (MFI 2006)}, Heidelberg, Germany, Sep. 2006, pp.
  512--517.

\bibitem{CDC07_HanebeckSchrempf}
U.~D. Hanebeck and O.~C. Schrempf, ``{G}reedy {A}lgorithms for {D}irac
  {M}ixture {A}pproximation of {A}rbitrary {P}robability {D}ensity
  {F}unctions,'' in \emph{Proceedings of the 2007 IEEE Conference on Decision
  and Control (CDC 2007)}, New Orleans, Louisiana, USA, Dec. 2007, pp.
  3065--3071.

\bibitem{ACC07_Schrempf-DiracMixt}
O.~C. Schrempf and U.~D. Hanebeck, ``{R}ecursive {P}rediction of {S}tochastic
  {N}onlinear {S}ystems {B}ased on {O}ptimal {D}irac {M}ixture
  {A}pproximations,'' in \emph{Proceedings of the 2007 American Control
  Conference (ACC 2007)}, New York, New York, USA, Jul. 2007, pp. 1768--1774.

\bibitem{MFI08_Hanebeck-LCD}
U.~D. Hanebeck and V.~Klumpp, ``{Localized Cumulative Distributions and a
  Multivariate Generalization of the Cram\'{e}r-von Mises Distance},'' in
  \emph{Proceedings of the 2008 IEEE International Conference on Multisensor
  Fusion and Integration for Intelligent Systems (MFI 2008)}, Seoul, Republic
  of Korea, Aug. 2008, pp. 33--39.

\bibitem{peacockTwoDimensionalGoodnessofFitTesting1983}
J.~A. Peacock, ``\BIBforeignlanguage{en}{Two-{{Dimensional Goodness}}-of-{{Fit
  Testing}} in {{Astronomy}}},'' \emph{\BIBforeignlanguage{en}{Monthly Notices
  of the Royal Astronomical Society}}, vol. 202, no.~3, pp. 615--627, Jan.
  1983.

\bibitem{fasanoMultidimensionalVersionKolmogorovSmirnov1987}
G.~Fasano and A.~Franceschini, ``\BIBforeignlanguage{en}{A {{Multidimensional
  Version}} of the {{Kolmogorov}}-{{Smirnov Test}}},''
  \emph{\BIBforeignlanguage{en}{Monthly Notices of the Royal Astronomical
  Society}}, vol. 225, pp. 155--170, Mar. 1987.

\bibitem{CDC09_HanebeckHuber}
U.~D. Hanebeck, M.~F. Huber, and V.~Klumpp, ``{Dirac Mixture Approximation of
  Multivariate Gaussian Densities},'' in \emph{Proceedings of the 2009 IEEE
  Conference on Decision and Control (CDC 2009)}, Shanghai, China, Dec. 2009.

\bibitem{JAIF16_Symmetric_S2KF_Steinbring}
J.~Steinbring, M.~Pander, and U.~D. Hanebeck, ``{The Smart Sampling Kalman
  Filter with Symmetric Samples},'' \emph{Journal of Advances in Information
  Fusion}, vol.~11, no.~1, pp. 71--90, Jun. 2016.

\bibitem{ACC13_Gilitschenski}
I.~Gilitschenski and U.~D. Hanebeck, ``{Efficient Deterministic Dirac Mixture
  Approximation},'' in \emph{Proceedings of the 2013 American Control
  Conference (ACC 2013)}, Washington D.C., USA, Jun. 2013.

\bibitem{stephensEDFStatisticsGoodness1974}
M.~A. Stephens, ``{{EDF Statistics}} for {{Goodness}} of {{Fit}} and {{Some
  Comparisons}},'' \emph{Journal of the American Statistical Association},
  vol.~69, no. 347, pp. 730--737, 1974.

\bibitem{fattalBlueNoisePointSampling2011}
R.~Fattal, ``Blue-{{Noise Point Sampling Using Kernel Density Model}},'' in
  \emph{{{ACM SIGGRAPH}} 2011 Papers}, ser. {{SIGGRAPH}} '11.\hskip 1em plus
  0.5em minus 0.4em\relax {New York, NY, USA}: {ACM}, 2011, pp. 48:1--48:12.

\bibitem{CISS14_Hanebeck}
U.~D. Hanebeck, ``{Kernel-based Deterministic Blue-noise Sampling of Arbitrary
  Probability Density Functions},'' in \emph{Proceedings of the 48th Annual
  Conference on Information Sciences and Systems (CISS 2014)}, Princeton, New
  Jersey, USA, Mar. 2014.

\bibitem{Fusion14_Hanebeck}
------, ``{Sample Set Design for Nonlinear Kalman Filters viewed as a Moment
  Problem},'' in \emph{Proceedings of the 17th International Conference on
  Information Fusion (Fusion 2014)}, Salamanca, Spain, Jul. 2014.

\bibitem{IFAC08_Huber}
M.~F. Huber and U.~D. Hanebeck, ``{Gaussian Filter based on Deterministic
  Sampling for High Quality Nonlinear Estimation},'' in \emph{Proceedings of
  the 17th IFAC World Congress (IFAC 2008)}, vol.~17, no.~2, Seoul, Republic of
  Korea, Jul. 2008.

\bibitem{ECC19_Li}
K.~Li, D.~Frisch, B.~Noack, and U.~D. Hanebeck, ``{Geometry-Driven
  Deterministic Sampling for Nonlinear Bingham Filtering},'' in
  \emph{Proceedings of the 2019 European Control Conference (ECC 2019)},
  Naples, Italy, Jun. 2019.

\bibitem{radonUberBestimmungFunktionen1917}
J.~Radon, ``{\"U}ber die {{Bestimmung}} von {{Funktionen}} durch ihre
  {{Integralwerte}} l{\"a}ngs gewisser {{Mannigfaltigkeiten}},'' \emph{Berichte
  {\"u}ber die Verhandlungen der K{\"o}niglich-S{\"a}chsischen Gesellschaft der
  Wissenschaften zu Leipzig. Mathematisch-Physische Klasse}, vol.~69, pp.
  262--277, 1917.

\bibitem{radonDeterminationFunctionsTheir1986}
------, ``On the {{Determination}} of {{Functions}} from {{Their Integral
  Values Along Certain Manifolds}},'' \emph{IEEE Transactions on Medical
  Imaging}, vol.~5, no.~4, pp. 170--176, Dec. 1986.

\bibitem{deansRadonTransformHigher1978}
S.~R. Deans, ``\BIBforeignlanguage{en}{The {{Radon Transform}} for {{Higher
  Dimensions}}},'' \emph{\BIBforeignlanguage{en}{Physics in Medicine and
  Biology}}, vol.~23, no.~6, pp. 1173--1175, Nov. 1978.

\bibitem{bonneelSlicedRadonWasserstein2015}
N.~Bonneel, J.~Rabin, G.~Peyr{\'e}, and H.~Pfister,
  ``\BIBforeignlanguage{en}{Sliced and {{Radon Wasserstein Barycenters}} of
  {{Measures}}},'' \emph{\BIBforeignlanguage{en}{Journal of Mathematical
  Imaging and Vision}}, vol.~51, no.~1, pp. 22--45, Jan. 2015.

\bibitem{strakaRandomizedUnscentedKalman2012}
O.~Straka, J.~Dunik, and M.~{\v S}imandl, ``Randomized {{Unscented Kalman
  Filter}} in {{Target Tracking}},'' in \emph{15th {{International Conference}}
  on {{Information Fusion}} ({{FUSION}} 2012)}, 2012, pp. 503--510.

\end{thebibliography}

\end{document}